\documentclass{article}

\usepackage{amsthm}
\usepackage{amsmath}
\usepackage{amssymb}
\usepackage{enumerate}
\usepackage{xcolor}
\usepackage{listings}
\usepackage{multirow}
\usepackage{color}
\usepackage[pdfencoding=auto, psdextra]{hyperref}
\usepackage{doi}
\usepackage{booktabs}
\usepackage{rotating}
\usepackage{bm}
\usepackage{tablefootnote}
\usepackage[normalem]{ulem}
\usepackage{array}
\usepackage{longtable}
\usepackage{bbm}

\newcommand{\E}{\mathbb{E}}
\newcommand{\ed}{\mathrm{d}}
\newcommand{\R}{\mathbb{R}}
\newcommand{\id}{{\bf 1}}
\renewcommand{\P}{\mathbb{P}}
\newcommand{\Q}{\mathbb{Q}}

\newcommand{\EZ}{\mathrm{EZ}}

\newcommand{\DT}{\ed {\cal T}(t)}

\newcommand{\1}{\mathbbm 1}

\newcolumntype{L}[1]{>{\raggedright\let\newline\\\arraybackslash\hspace{0pt}}p{#1}}

\DeclareMathOperator{\lcm}{lcm}

\DeclareMathOperator{\argmax}

\newtheorem{theorem}{Theorem}
\newtheorem{lemma}[theorem]{Lemma}
\newtheorem{remark}[theorem]{Remark}

\newtheorem{example}[theorem]{Example}

\newtheorem{inneraxiom}{Axiom}
\newenvironment{axiom}[1]
{\renewcommand\theinneraxiom{#1}\inneraxiom}
{\endinneraxiom}
\theoremstyle{definition}
\newtheorem{definition}[theorem]{Definition}
\numberwithin{equation}{section}
\numberwithin{theorem}{section}

\title{Asymptotically Optimal Management of Heterogeneous Collectivised Investment Funds}
\author{John Armstrong, Cristin Buescu}

\begin{document}
		
\maketitle

\begin{abstract}
    A collectivised fund is a proposed form of pension investment, in which
    all investors agree that any funds associated with deceased members should be split
    among survivors.
    For this to be a viable financial product, it is necessary to know how to manage
	the fund even when it is heterogeneous: that is when different investors
	have different preferences, wealth and mortality. There is no obvious way to define a
	single objective for a heterogeneous fund, so this is not an optimal control problem.
	In lieu of an objective function, we take an axiomatic approach. Subject to our axioms
	on the management of the fund, we find an upper bound on the utility that can be achieved
	for each investor, assuming a complete markets and the absence of systematic longevity risk. We give a
	strategy for the management of such heterogeneous funds which achieves this bound asymptotically
	as the number of investors tends to infinity.
\end{abstract}	

\section*{Introduction}

We study the problem of maximizing the benefit one can obtain from one's pension
if one is willing to pursue a collective strategy. This is a strategy in which a group of individuals
agree that all assets left by an individual who dies are shared among the survivors.

By collectivising investments one counters idiosyncratic
longevity risk. Unlike an annuity which guarantees a fixed real-terms income until retirement,
a collectivised fund may pursue a risky strategy to take advantage of the equity risk-premium
and so yield a higher utility for investors. Additionally, a fund may also exploit
intertemporal substitution of consumption, by delaying consumption until investments have grown.
These considerations imply that a collectivised fund should out-perform both an annuity and a individually
managed pension funder. For a more detailed discussion of the benefits of collectivised pension investment, see \cite{ab-main}.

We will model a collective fund of $n$ individuals who have already retired and invested their
capital in the fund. We will then consider the optimal pattern of investment and consumption.

For simplicity we assume that the only stochastic risk factors
are given by the market and idiosyncratic mortality risk. Our modelling for the market, mortality and individual 
preferences is described in detail in \ref{sec:marketAndMortality}. We will also give a number of remarks on how systematic longevity risk could be included in the analysis.

In Section \ref{sec:homogeneousModel} we give a mathematical description of the optimal investment problem
for a homogeneous fund: that is to say, a collective fund where all the individuals are isomorphic. We also show that
an overall objective function for the fund can be derived
using the ideas of robust
optimization or the notion of ``The Veil of Ignorance'' introduced in \cite{rawls}. Both approaches yield the same objective function. The
management of the fund may then be written as an optimal control problem. We state the optimal control problems for both the cases of continuous and discrete time consumption.
We do not solve these problem in this paper. The discrete time problem is
solved analytically in the case of homogeneous Epstein--Zin preferences in the Black--Scholes model in \cite{ab-ez} building
on the techniques of \cite{merton1969lifetime,campbellViceira}.
A numerical algorithm for exponential Kihlstrom--Mirman preferences is given in \cite{ab-exponential}.

In Section \ref{sec:basicProperties} we prove some basic properties of collectivised pension investment. Specifically
we prove that, under very mild assumptions, collectivisation is always beneficial. We also give a sufficient condition
for constant consumption to be the optimal strategy. These conditions are extremely stringent and indicate that an annuity will typically be a suboptimal investment.

The mathematical novelty in the paper appears in Section \ref{sec:heterogeneous} where we consider how
to manage heterogeneous funds. These are funds where preferences, capital and mortality distributions vary between investors. This problem is challenging because it is not obvious how to combine each individual's preferences to
obtain a single objective for the fund as a whole. Rather than attempt this, we take an axiomatic approach
describing the properties that any management strategy for the fund should possess. From these axioms we are
able to deduce an upper bound on the utility of each investor in a heterogeneous fund, on the assumption that
the market is complete. Given the times-scale of pension investments this assumption is a very reasonable
approximation to reality. We are then able to describe a strategy that asymptotically achieves this bound as
the number of investors tends to infinity.

\section{Market, preference and mortality models}
\label{sec:marketAndMortality}

\subsection{Market model}

We will assume that a fund may invest in a market determined by
a filtered probability space $(\Omega^M, {\cal F}, {\cal F}^M_{t\in \R^+}, \P^M)$.
In particular we will only consider continuous time markets.

We assume
there are $k$ available assets and that the price of asset $j$ at time $t$ in real terms is given by $S^j_t$. Arbitrary quantities of assets can be bought or sold
at these prices. We assume that the asset $1$ is risk free,
so that $\ed S^1_t=r_t S^1_t \ed t$
where $r_t$ is the short-rate.
We assume as our no-arbitrage condition that there there is an equivalent measure $\Q^M$ such that $(S^0_t)^{-1}S^i_t$ is a $\Q$-Martingale
for each $i$.
We will call such a market an {\em infinitely-liquid} market.

\subsection{Mortality model}

We will later assume that consumption only occurs at times in a set ${\cal T}$ which
may be either $[0,T)$ or the evenly spaced time grid $\{ 0,\, \delta t,\,  2 \delta t,\,  3 \delta t, \ldots, T-\delta t \}$ where
$T$ is an upper bound on an individual's possible age which may be infinite. 
As a result we may assume that mortality events are also restricted to ${\cal T}$.
We write $\DT$ for the measure determined by the index set: this will be the Lebesgue measure on $[0,\infty)$ in the continuous case, or the sum of Dirac masses of mass $\delta t$ at each point in ${\cal T}$ for the discrete case. 

We assume that the random variables $\tau_i$ representing the time of death of individual $i$
are independent and absolutely continuous with respect to $\DT$, with distribution given by $p^i_t \, \DT$. We will 
write $(\Omega^{L}, {\cal F}^L, {\cal F}^{L}_t, \P^L)$ for the filtered probability
space generated by all the $\tau_i$, the filtration is obtained by requiring that each $\tau_i$ is a stopping time. We will write $F_\tau(t)$ for the distribution function of $\tau$.

The assumption of independence means that we are only considering idiosyncratic longevity
risk, but we will give some remarks on how systematic longevity risk will affect our findings later in the paper.

We will write $n_t$ for the number of survivors at time $t$, that is the number of individuals whose time of death is
greater than or equal to $t$.
This convention ensures that $n_0=n$ and works well with our convention that cashflows
received at the time of death are still consumed.
Note, however, that $n_{t+\delta t}$ will be ${\cal F}^L_t$ measurable.

\subsection{Preference model}

We assume that there are $n$ individuals. We model a ``pension outcome'' for individual $i$ as a pair $(\gamma^i,\tau_i)$ consisting of a stochastic process
$\gamma^i_t$, representing the rate of pension payments to individual $i$ at time $t$, and
the random variable $\tau_i$ representing the time of death. The underlying
filtered probability space will be denoted by $(\Omega, {\cal F}, {\cal F}_t, \P)$ and is assumed
to satisfy the usual conditions. The units
of $\gamma_t$ should be taken to be in real terms which ensures that our models
for inflation and preferences are separate.

We will later assume that consumption only occurs at times in a set ${\cal T}$.
It will occasionally be convenient to allow the cashflow $\gamma_t$ to be non-zero when $t>\tau$, but this cash will not be consumed.
In the discrete case we assume that cashflow at the moment of death $\gamma_\tau$ is still consumed. So the total consumption
over the lifetime of an individual is
\[
\int_0^\tau \gamma_t \, \DT.
\]

We will assume that the preferences of our individuals are
given by a $\R \cup \{ \pm \infty \}$ valued {\em gain function} ${\cal J}_i(\gamma,\tau_i)$ which acts on pension outcomes. The individual
prefers pension outcomes which yield higher values of the gain function. We will assume that the gain function takes the value $-\infty$
whenever $P(\gamma_t \id_{t\leq \tau}(t) < 0)>0$. Following \cite{ab-main} we make the following definitions.

\begin{definition}
The preferences are said to be {\em invariant} if they are invariant under filtration preserving automorphisms of the
probability space.
\end{definition}
\begin{definition}
We will say that ${\cal J}_i$ is {\em concave} if it is concave as a function of $\gamma$ for all $\tau_i$.
\end{definition}
\begin{definition}
The preferences are {\em monotonic} if ${\cal J}_i(\gamma, \tau) \leq {\cal J}_i(\gamma^\prime, \tau)$ if $\gamma_t \leq \gamma^\prime_t$
for all $t \in {\cal T}$.
\end{definition}
\begin{definition}
A gain function {\em does not saturate} if whenever
${\cal J}(\gamma, \tau)$ is finite we have $J(\gamma+\epsilon,\tau)>
{\cal J}(\gamma, \tau)$ for all positive $\epsilon$.
\end{definition}

\begin{example}
	\label{def:vnm}	
	{\em Von Neumann--Morgernstern preferences with mortality} are determined by a
	choice of concave, increasing utility function $u:\R_{\geq 0} \to \R$ and a discount rate $b$.
	The gain function is given by
	\[
	{\cal J}(\gamma,\tau) =
	\E\left( \int_0^\tau e^{-bt} u( \gamma_t )\, \DT \right).
	\]
\end{example}

\begin{example}
	{\em Exponential Kihlstrom--Mirman preferences with mortality} are determined by a
	choice of concave, increasing utility function $u:\R_{\geq 0} \to \R$.
	The gain function is given by
	\[
	{\cal J}(\gamma,\tau) =
	\E\left( -\exp\left( -\int_0^\tau u( \gamma_t )\, \DT \right) \right).
	\]
\end{example}

Both of these gain functions give rise to concave, invariant, monotonic preferences. As is
explained in detail in \cite{ab-main}, these gain functions are also {\em stationary} (meaning
that the preferences do not depend upon the time period being considered)
and {\em law-invariant} (meaning that the preferences depend only on the probability law
of the consumption and not the time at which information is received).

A larger class of stationary preferences over consumption may be obtained if one drops the requirement
that preferences are law-invariant. Such preferences were studied in discrete time 
by Kreps and Porteus, \cite{krepsPorteus}, and the particular case of Epstein--Zin preferences, introduced
in \cite{epsteinZin1}, has proved popular in applications as they allow for
separate risk-aversion and intertemporal substitution parameters while maintaining
mathematical tractability. Such preferences
have been used to resolve a number of asset pricing
puzzles, see for example \cite{bansalYaron,bansal,benzoniEtAl,bhamraEtAl}.

In \cite{xing}, a continuous time analogue of Epstein--Zin preferences
is defined
using BSDEs on probability spaces equipped with a filtration generated by Brownian motion, building on the continuous time analogues of Kreps--Porteus preferences
introduced in \cite{duffieEpstein}. A detailed literature review of the development of such preferences
is given in \cite{xing}. Motivated by \cite{xing} we may define continuous-time
Epstein--Zin preferences with mortality as follows.

\begin{example}
Suppose we are working with continuous time
mortality. Suppose that $\Omega=\Omega^M \times \Omega^L$ and that the filtered probability
space $\Omega^M$ is generated by d-dimensional Brownian motion. Let
$N^i_t={\1}_{\tau^i \leq t}$ be the jump process associated with
the death of individual $i$. Let $\Lambda^i_t$ be the predictable compensator
of $N^i_t$ given by
\begin{equation}
    \Lambda^i_t = \int_0^{t \wedge \tau^i} \lambda_s \, \ed s.
\end{equation}
where $\lambda_t$ is the force of mortality.
Let $M^i_t$ be the compensated martingale process defined by $M^i_t=N^i_t-\Lambda^i_t$.

Let ${\cal S}_2$ denote the set of $\R$
valued progressively measurable processes $\{Y_t\}$ $(0\leq t\leq T)$ such that
\[\|Y\|^2_{{\cal S}_2}:=\E[ \sup_{t\geq 0} |Y_{t \wedge T}|^2]<\infty.\]
Let $L^2(W)$ denote the set of $\R^d$
valued progressively measurable processes $Z_t$ $(0\leq t\leq T)$ such
that
\[ \|Z\|^2_{L^2(W)}:=\E[ \int_0^\infty \|Z_t\|^2 \ed t ]<\infty.
\]
Let $L^2(\lambda)$ denote the set of $\R^n$ valued progressively measurable
processes $\zeta_t$ $(0\leq t\leq T)$ such that
\[
\|\zeta\|^2_{L^2(\lambda)}:=\E[ \sum_{i=1}^n \int_0^{t \wedge \tau^i} |\zeta_s^i|^2 \lambda_s \,\ed s ] < \infty.
\]

Let $b>0$ be a discount rate. Let $0\neq\rho<1$ and
$0\neq\alpha<1$ be parameters determining the individual's
satiation and risk preferences.
The Epstein--Zin aggregator $f:[0,\infty) \times (-\infty,0] \to \R$ is defined by
\[
f(\gamma,v):=b \frac{ \alpha v }{\rho}
\left(
\left(
\frac{\gamma}{(\alpha v)^\frac{1}{\alpha}}
\right)^\rho - 1
\right).
\]
We will assume further that $\alpha<0$ and $0<\rho<1$.
These parameter restrictions are justified in \cite{xing} 
\footnote{
\cite{xing} uses a slightly different parameterization. Their
parameters $\delta$, $\gamma$ and $\psi$
are related to ours by $\delta=b$, $\gamma=1-\alpha$ and $\psi=\frac{1}{1-\rho}$.
}
on the grounds of empirical relevance.
We require one additional parameter $a>0$
which we call the {\em adequacy level}.

A stopping time $\tau$ is admissible if $\tau < T$
almost surely.
Given an admissible stopping time $\tau$, the set of admissible consumption streams is defined by
\[
{\cal C}:=\{ \gamma \in {\cal R}_+
\mid \E[ \int_0^\tau e^{-b s} \gamma_s \,\ed s]<\infty
\}
\]
where ${\cal R}_+$ is the set of all non-negative
progressively measurable processes.
Given an admissible consumption and mortality $(\gamma_t,\tau)$
we define $(V_t,Z_t,\zeta_t)$ to be the solution of the BSDE
\begin{equation}
\ed V_t =
f(\gamma_t, V_t) {\1}_{t \leq \tau} \, \ed t
- Z_t \, \ed W_t
- \sum_{i=1}^n \zeta_t^i \, \ed M^i_t, \quad 0 \leq t \leq T; \quad V_T = U
\label{eq:bsde1}
\end{equation}
where $(V,Z,\zeta) \in {\cal S}_2 \times L^2(W) \times L^2(\tilde{\lambda})$ and where we set $U=\frac{a^\alpha}{\alpha}.$
The existence and uniqueness of solutions of the
BSDE \eqref{eq:bsde1} may be established using the technique of the 
proof of Proposition 2.2 
of \cite{xing} combined with the properties
of BSDEs with default jumps established in \cite{dumitrescu}.
Since we assume that $\tau$ is bounded above by $T$,
the techniques of \cite{aurandHuang} are not
required.

We may then define the Epstein--Zin utility with mortality for adequacy level $a$ associated with $(\gamma,\tau)$
by
\[
{\cal J}(\gamma,\tau):=V_0.
\]
\end{example}

\bigskip

In the classical Epstein--Zin utility without
mortality, the term $U$ in \eqref{eq:bsde1}
is a ${\cal F}_\tau$-measurable random variable given by a $U=\frac{c_\tau^\alpha}{\alpha}$ where $c_\tau$ denotes a bequest at time $\tau$.
Such a formulation is attractive mathematically as it
ensures the preferences are positively homogeneous. This
symmetry then allows one to reduce the dimension
of control problems involving such preferences.

However, we note that when $\alpha<0$, this modelling choice would associate an infinitely negative utility
to dying without a bequest. No amount of lifetime consumption would be sufficient to overcome this. This does not seem a plausible model for typical pension preferences.

In our model, we call $a$ an adequacy level because the
outcome $(\gamma,\tau)$ has the same utility as
$(\tilde{\gamma},T)$ where $\tilde{\gamma}$
is equal to $\gamma$ up to time $\tau$ and equal to
$a$ thereafter (this is because $f(a,U)=0$). Thus an investor is indifferent to dying or living at the 
adequacy level. 
The importance of the adequacy level in pension
modelling is studied in \cite{ab-main}.

The techniques of \cite{xing} allow one to
prove that Epstein--Zin preferences with mortality are concave and monotone. 

\begin{lemma}
Epstein--Zin preferences with mortality are invariant.	
\end{lemma}	
\begin{proof}
Let $\phi:\Omega \to \Omega$ be a filtration automorphism. Let $(V,Z,\zeta) \in {\cal S}_2 \times L^2(W) \times L^2(\tilde{\lambda})$
be the solution to \ref{eq:bsde1}. Then we have
\[
\ed (V\circ\phi)_t =
f((\gamma\circ\phi)_t, (V\circ\phi)_t) {\1}_{t \leq \tau} \, \ed t
- (Z\circ\phi)_t \ed (W\circ \phi)_t
- \sum_{i=1}^n (\zeta\circ\phi)^i_t \, \ed (M^i\circ\phi)_t
\]
for all $0 \leq t \leq T$ and $\quad V_T\circ\phi = U$. By the martingale
representation theorem for processes with default jumps (see e.g.\ \cite{jeanblancEtAl}) we may find
$\tilde{Z}_t \in L^2(W)$ and $\tilde{\zeta} \in L^2(\tilde{\lambda})$
such that
\[
\ed (V\circ\phi)_t =
f((\gamma\circ\phi)_t, (V\circ\phi)_t) {\1}_{t \leq \tau} \, \ed t
- \tilde{Z}_t \ed W_t
- \sum_{i=1}^n \tilde{\zeta}_t^i \, \ed M^i_t,\quad 0 \leq t \leq T.
\]
Hence the Epstein--Zin utility with mortality for $\gamma\circ\phi$ is equal to $(V\circ\phi)_0=V_0$ (the last
equality follows since ${\cal F}_0=\{\emptyset,\Omega\}$ and so $V(\omega)_0$ is independent of $\omega$).
\end{proof}

\section{Managing homogeneous funds}
\label{sec:homogeneousModel}

In this section we consider how
to manage homogeneous funds. We call a fund homogeneous if all individuals in the fund have identical
preferences, wealth and mortality.  We will consider the management of
heterogeneous funds in \ref{sec:heterogeneous}.

Since we assume that each individual has an identical mortality distribution we may
define $p_t:=p^i_t$.

Let us first consider the case of finite $n$.

We wish to decide how
to manage a collective pension fund where individual contributes an amount $X_0$ at time $0$. Individual $i$ will receive an income $\gamma^i_t$ at
time $t$, with $\gamma^i_t=0$ if the individual is dead at that time. Any cash that is yet to be consumed is invested in the market. There is no bequest when an individual dies, all remaining cash is shared with the fund.

\bigskip

We need to choose an objective for the fund itself. We informally outline two
proposals which we will then explain in more detail.
\begin{enumerate}[(i)]
	\item  Suppose that the individual gain function is law-invariant
		   and so depends only on the distribution of outcomes. We may
		   understand ``distribution'' to mean distribution in both probability and distribution across the population. In this
		   way the individual gain function gives rise to  an objective for the entire fund. We will call this
		   the ``distribution approach''.
	\item  We follow the robust optimization approach to managing
		   the fund: we maximize the infimum of the individual gain functions.
\end{enumerate}	

Let us explain the mathematical detail required for the distribution approach.
We suppose that the individual gain function is law-invariant. We define a discrete uniformly distributed random variable $\iota$ 
which takes values in $\{1, \ldots, n\}$.
We write $(\Omega^\iota, \sigma^\iota, \P^\iota)$ for the probability space generated by $\iota$. We define a filtration ${\cal F}^\iota_{t \in {\R^+ \cup \{ \infty \}}}$ by
\[
{\cal F}_t^\iota = \begin{cases}
\{ \Omega^\iota, \emptyset \} & t < \infty \\
\sigma^\iota & t= \infty.
\end{cases}
\]
Thus $\Omega^\iota$ represents a random choice of individual made at time $\infty$.
If we have a law-invariant individual gain function ${\cal J}_\iota$ defined relative
to a probability space $\Omega$, we can then define a gain function relative to $\Omega \times \Omega^\iota$ by requiring
\begin{equation}
{\cal J}^D(\gamma)={\cal J}_\iota(\gamma^\iota, \tau_\iota).
\label{eq:distributionObjective}
\end{equation}
Note that since $\tau_\iota$ is not a stopping time, this gain function
can only be given a meaning for law-invariant individual gain functions
${\cal J}_\iota$.

The gain function for the robust approach is given by
\begin{equation}
{\cal J}^R(\gamma):=\inf_{i \in I_0} {\cal J}_i(\gamma^i, \tau_i).
\label{eq:robustObjective}
\end{equation}

These two approaches are not equivalent in general. Consider the Biblical problem faced by Solomon of distributing a child among two women who claim to be its mother.
In the distribution approach, giving the child to a randomly selected woman would be optimal. In the robust approach, giving neither woman the child would be an equally optimal alternative. For concave individual gain functions, Solomon's recommended approach of splitting the child in two would be ideal.

Having decided on a gain function ${\cal J}$ for our fund, we can write down the associated optimization problem.

We may also augment our probability space with
a filtered probability space $\Omega^\perp$ so 
that any arbitrary decisions that need to be made (such as choosing a random woman to give the child) can be made using random variables defined on this space. We then take $\Omega = \Omega^M \times \Omega^L \times \Omega^\iota \times \Omega^\perp$ equipped with the product
filtration ${\cal F}_t$ and product measure $\P$. As one might expect, under reasonable
conditions, $\Omega^\perp$ proves to be irrelevant, and the optimal strategies can be taken
to be $\Omega^M \times \Omega^L \times \Omega^\iota$ measurable. See Remark \ref{remark:omegaPerp} below. Note that when working with Epstein--Zin preferences we will always assume $\Omega^\perp$ is trivial to ensure that the preferences are defined.

We then wish to choose progressively measurable consumption streams $\gamma^i_t \geq 0$ with $t \in {\cal T}$ and investment proportions
$\alpha^j_t$ in asset $j$ with $t \in \R^+$.
We require $\sum_{j=1}^k \alpha^j_t\leq 1$.
We have an inequality in this  equation, because it is
acceptable, if sub-optimal, to simply discard assets.
We must choose these processes
such that the total wealth of the fund is always non-negative.
We write
${\cal A}$ for the resulting set of admissible controls $(\bm{\gamma},\bm{\alpha})$,
and we will now give  a fully precise mathematical description of this set.

First let us write down the dynamics of the fund value. For continuous time consumption the fund value 
satisfies the SDE
\begin{equation}
\ed F_t = \sum_{i=1}^k \alpha^i_s F_s \ed S^i_s - \sum_{i=1}^n \gamma^i_t \, \ed t
\label{eq:fundValueCts}
\end{equation}
with $F_0=\sum_{i=1}^n B_i$, where $B_i$ is the initial budget
of individual $i$ (in this section we are assuming that all individuals
have identical initial budget $B_i=X_0$, but the equations for the more general case
will be useful later). For discrete
time consumption, let $F_t$ denote the fund value before consumption and
$\overline{F}_t$ denote the fund value after consumption. We then
have the following budget equations for the dynamics of $F_t$ and $\overline{F}_t$.
\begin{equation}
\begin{split}
F_t &= \begin{cases}
\sum_{i=1}^n B_i & \text{t = 0} \\
\lim_{h\to 0+}\overline{F}_{t-h} &  t \in {\cal T}\setminus\{0\} \\
\overline{F}_t & \text{otherwise.}
\end{cases} \\
\overline{F}_t &=
\begin{cases}
F_t - \sum_{i=1}^n \gamma^{i}_t  \\
\overline{F}_{t^\prime} + \sum_{i=1}^k \int_{t^\prime}^t \alpha^i_s \overline{F}_s \, \ed S^i_s
& t^\prime \in {\cal T} \text{ and } t^\prime\leq t < t^\prime+\delta t.
\end{cases}
\end{split}
\label{eq:fundValue}
\end{equation}
A strategy is $(\bm{\gamma},\bm{\alpha})$ is admissible if it ensures $F_t \geq 0$ and $\overline{F}_t \geq 0$ at all times. Hence
\begin{equation}
{\cal A}=\{(\bm{\gamma},\bm{\alpha})\in {\cal PM} \mid F_t \geq {0} \text{ and } \overline{F}_t \geq 0, \forall t \}
\label{eq:admissibleControls}
\end{equation}
where ${\cal PM}$ is the set of progressively measurable $\R^n \times \R^k$ valued processes for the probability space $\Omega$.

Our objective is to compute
\begin{equation}
v_n = 
\sup_{(\bm{\gamma},\bm{\alpha}) \in {\cal A}} {\cal J}(\gamma).
\label{eq:fundObjective}
\end{equation}
and to find $(\bm{\gamma}, \bm{\alpha})$ achieving this supremum.

We will henceforth assume that our individual gain functions are concave.
Since the market is positively homogeneous, given a strategy $(\bm{\gamma},\bm{\alpha})$
we may form a new strategy $(\bar{\bm{\gamma}},\bm{\alpha})$ which assigns
the mean consumption to all survivors:
\[
\bar{\gamma}^i_t =
\begin{cases}
0 & \tau_i < t \\
\frac{1}{n_t} \sum_{j=1}^n \gamma^j_t & \text{otherwise}.
\end{cases}
\]
The concavity of our individual gain functions ensures that for
our fund's gain function we have
\[
{\cal J}(\bar{\gamma}^i,\alpha) \geq {\cal J}(\gamma^i,\alpha).
\]
So we may assume without loss of generality that all survivors consume the same amount $\gamma_t$ at time $t$. Under this assumption we have that ${\cal J}^D={\cal J}^R$, so the distinction between the distributional approach
and the robust approach will not in fact be important.

Our motivation for introducing the two approaches is that we
believe the distributional approach corresponds more closely to the
intuitive notion of optimal investment for a collective fund, however the
robust approach is required if we wish to use Epstein--Zin utility
as the individual gain function.
To justify our claim that the distributional approach is more intuitive we first note the example of Solomon above. A second justification is
given by the concept of the ``Veil of Ignorance'' described in \cite{rawls}.
This concept suggests that when making collective decisions one should
make those decisions as though the identity of the individuals was unknown.
Our filtered probability space $\Omega^\iota$ represents this veil of ignorance: the veil being lifted at time $\infty$, but the control
$(\gamma, \alpha)$ is chosen in a state of ignorance.

\bigskip

Having decided that the consumption, $\gamma_t$, should be the same
for all individuals, we may now consider how to model
infinite collectives ($n=\infty$). When performing accounting
calculations in this case, we will perform all calculations on a
per-individual basis. For example, rather than keep track of the total fund value which would be infinite, we keep track of the fund value per individual which will be finite. We will assume that a deterministic
proportion of the original individuals dies over each time interval given by the integral of $p_t \, \DT$ over that interval. 
This assumption allows us to include mortality within our accounting.

Let us express this precisely.
Let $Y_t$ represent the fund
value per individual at time $t$ before consumption or mortality,
and let $\overline{Y}_t$ represent the fund value per individual after consumption. Then in the continuous time case we have
\begin{equation}
\ed Y_t = \sum_{i=1}^k \alpha^i_s Y_s \ed S^i_s - \pi_t \gamma_t \, \ed t
\label{eq:fundValuePerIndividualCts}
\end{equation}
where $\pi_t$ is the proportion of individuals surviving to time $t$.
In the 
discrete time case we have
\begin{equation}
\begin{split}
Y_t &= \begin{cases}
X_0 & \text{t = 0} \\
\lim_{h\to 0+}\overline{Y}_{t-h} &  t \in {\cal T}\setminus\{0\} \\
\overline{Y}_t & \text{otherwise.}
\end{cases} \\
\overline{Y}_t &=
\begin{cases}
Y_t - \pi_t \gamma_t  & t \in {\cal T} \\
\overline{Y}_{t^\prime} + \sum_{i=1}^k \int_{t^\prime}^t \alpha^i_s \overline{Y}_s \, \ed S^i_s
& t^\prime \in {\cal T} \text{ and } t^\prime\leq t < t^\prime+\delta t.
\end{cases}
\end{split}
\label{eq:fundValuePerIndividual}
\end{equation}
For the case $n=\infty$, equations \eqref{eq:fundValuePerIndividualCts}
and \eqref{eq:fundValuePerIndividual} define the process $Y_t$. In
the case $n=\infty$, $\pi_t=1-F_\tau(t)$ is deterministic. For
the case of finite $n$, equations \eqref{eq:fundValuePerIndividualCts}
and \eqref{eq:fundValuePerIndividual} follow from
\eqref{eq:fundValueCts} and \eqref{eq:fundValue}. In the
case of finite $n$, $\pi_t=\frac{n_t}{n}$ is a random variable.

In the case $n=\infty$ we take as the gain function for our fund
\begin{equation}
{\cal J}(\gamma):={\cal J}_1(\gamma, \tau_1).
\label{eq:robustObjectiveInfinite}
\end{equation}
This is reasonable since we have assumed that all living
individuals receive the same consumption stream $\gamma_t$
and have isomorphic preferences. Alternatively
if the individual gain function is law-invariant we may define
a random variable $\tau_\iota$ which is measurable only at $\infty$
and which has distribution $p_t \ed t$. We may then write the
gain function as
\begin{equation}
{\cal J}(\gamma):={\cal J}_\iota(\gamma, \tau_\iota).
\label{eq:distributionObjectiveInfinite}
\end{equation}
These two formulations will be equivalent, but
\eqref{eq:distributionObjectiveInfinite} seems a more intuitive
formulation.

We define the set ${\cal A}$ of admissible strategies
in the case $n=\infty$ by saying that a
strategy is admissible if $Y_t \geq 0$ for all time.
We may now define $v_\infty$ via equation \eqref{eq:fundObjective}
as before.

We note that our approach to treating of the case $n=\infty$ is
at present rather heuristic,
but it will be justified rigorously via limiting arguments in the next section.

\section{Basic properties of collectivised investment}
\label{sec:basicProperties}

\subsection{Collectivisation is beneficial}

It is intuitively clear that collectivisation should be beneficial. We give a formal
proof within our model.

The probability spaces we have defined depend upon the number of individuals, $n$, but if
$n<m$ there is an obvious way to map random variables defined on the probability space
for $n$ individuals to random variables defined on the probability space for $m$ individuals that preserves the random variables defining the market and the mortality of the first $n$ individuals. We shall assume henceforth that the individual gain functions are preserved under this mapping. This will happen automatically if we define our gain functions using a specification such as Epstein--Zin preferences with mortality with parameter values $(\alpha,\rho,b,a)$.

\begin{theorem}
	\label{thm:increasingValue}	
	If the individual gain functions ${\cal J}_i$ are
	concave then
	\[
	v_n \leq v_m \quad \text{if }n\leq m<\infty.
	\]
	Moreover $v_n \leq v_\infty$ in complete markets.
	In place of assuming that the market is complete, one may instead
	require that admissible strategies are integrable.
\end{theorem}
\begin{proof}
	We recall that, by assumption, the individual gain functions
	${\cal J}_i$ are invariant and isomorphic. We note that to each permutation of the individuals we can associate an automorphism of the probability space which permutes the individuals. Hence a strategy
	which is effective for one set of $n$ individuals will 
	give rise to an isomorphic strategy for another set of $n$ individuals.		
	
	Suppose that $m$ is finite. Suppose we set up $\binom{m}{n}$
	funds corresponding to each possible choice of $n$ individuals from the
	full set of $m$ individuals and allocate the initial budget equally to each of these funds.  Given an admissible strategy which achieves a value $v$ for the gain function for the first $n$ individuals,
	we can use it in each of these
	$\binom{m}{n}$ funds. By the concavity of the gain function, we see
	that the resulting strategy will have a value greater than or equal to $v$. The result follows.
	
	Now consider the case when $m=\infty$. Let
	$(\bm{\gamma},\bm{\alpha})$ be an admissible strategy
	for a collective of $n$ investors.
	Our budget constraint
	together with the requirement that the
	discounted asset prices are $\Q$-measure
	martingales ensures that
	\[
	\E_{\Q\times \P^L}\left( \frac{1}{n} \sum_{i=1}^n \int \frac{\gamma^i_t}{S^1_t} \, \DT \right) \leq X_0.
	\]
	Note that $\gamma^i_t$ is non-negative. Hence by Fubini's theorem we may define a stochastic process
	$\overline{\gamma}_t$ by
	\[
	\frac{\overline{\gamma}_t}{S^1_t}:=\E_{\P^L}\left( \frac{1}{n} \sum_{i=1}^n \int \frac{\gamma^i_t}{S^1_t} \, \DT \mid \bm{S}_t \right) \leq X_0.
	\]
	where ${\bm{S}}_t$ is the vector of asset prices at time $t$.
	This will be progressively measurable with
	respect to the filtered probability space of the market
	$(\Omega^M)$ and moreover will satisfy
	\[
	\E_{\Q}\left( \int \frac{\overline{\gamma}_t}{S^1_t} \, \DT \right) \leq X_0.
	\]	
	It follows from our assumption that the market is complete that
	we can then find $\overline{\bm{\alpha}}$ which ensures that
	$(\overline{\gamma}, \overline{\bm{\alpha}})$ is an admissible strategy
	for an infinite collective.
	
	The concavity of the gain function now ensures that
	\[
	{\cal J}_i(\gamma^i, \tau_i) \leq {\cal J}_{\iota}(\overline{\gamma}, \tau_\iota).
	\]
	Hence $v_\infty \geq v_n$.	
	
	If the market is not complete, we instead use the integrability of ${\bm{\alpha}}$
	in order to define $\overline{\bm{\alpha}}$ by averaging.
\end{proof}

To prove that $v_n \to v_\infty$ as $n \to \infty$ we require a little more than concavity and invariance, one also needs some degree of continuity in the gain function when it is perturbed
by a low probability event.
As an example we prove the following in Appendix \ref{sec:ezConvergence}.

\begin{theorem}
If preferences are given by Epstein--Zin preferences with mortality with $0<\rho<1$
and $\alpha<0$ then $v_n \to v_\infty$ as $n \to \infty$.
\label{thm:ezConvergence}
\end{theorem}

See \cite{ab-exponential} for the case of exponential utility.

\begin{remark}
The proof of Theorem \ref{thm:increasingValue} will continue to work if one introduces systematic longevity risk. The main challenge in extending \ref{thm:ezConvergence} lies in defining
the utility function in a richer probability space. Nevertheless, the continuity required
to establish \ref{thm:ezConvergence} can be expected for any gain function that yields a
well-posed optimization problem.
\end{remark}

\subsection{When is constant consumption optimal?}

Defined benefit pensions and annuities typically aim to provide constant consumption stream in real terms.
It is therefore natural to ask when constant consumption is optimal. Our
next result gives sufficient conditions. We assume for technical reasons
that the asset price processes $S^i_t$ are all given by diffusion processes. We will call
a market that satisfies this assumption a diffusion market.

\begin{theorem}
	Constant consumption is optimal for von Neumann--Morgernstern
	preferences with $b=0$
	in complete diffusion markets with $\P=\Q$ and $r_t=0$ when $n=\infty$. If $p_t=0$ for all times
	except $T-\delta t$ then constant consumption is optimal for all $n$.
	For example, in a Black--Scholes--Merton market with no drift and risk free rate
	of zero one has $\P=\Q$.
	\label{thm:constantConsumption}
\end{theorem}
\begin{proof}	
	We will give our proof using symmetry arguments rather than
	direct calculation. Our aim in using this approach is to explain
	why this result feels ``obvious''.	
	
	First consider the case $n=\infty$.
	
	By the classification of standard probability spaces, the filtered probability space $(\Omega, {\cal F}_{t \in {\cal T}}, \P)$ is isomorphic to the Cartesian product 
	$(S^1)^{\cal T}$ where $S^1$ is the circle of circumference $1$. Note that we are using the assumption that we are in a diffusion market to ensure that these probability
	spaces are standard and atomless.
	Thus rotations of the circles give rise to market isomorphisms (see \cite{armstrongClassification} for a definition
	of market isomorphism). This gives an action of the Lie group $(S^1)^{\cal T}$ on our market and hence on the space of admissible investment strategies. Given any strategy we may apply these rotations
	to obtain new, equivalent strategies. Choose an invariant metric on our Lie group, so that we may define the average strategy.
	Note that we are using market completeness here, just as we did in the proof of Theorem \ref{thm:increasingValue}, to ensure
	that the average strategy exists.
	By the concavity of our maximization problem and Jensen's inequality, this averaged
	strategy will outperform the original strategy. But the averaged strategy will be invariant under the group action and hence deterministic.
	
	Since we may assume our strategy is deterministic, it will be a pure bond investment so $Y_{t+\delta t}=\tilde{Y}_t$. From our dynamics \eqref{eq:fundValuePerIndividual} we have
	\[
	Y_{t+\delta t} = Y_t - (1-F_{\tau}(t))\gamma_t
	\]
	Solving this difference equation yields the (intuitively obvious) budget constraint
	\begin{equation}
	0 \leq X_0 = \int_0^T (1-F_{\tau}(t)) \gamma_t \, \DT.
	\label{eq:vnmBudgetConstraint}
	\end{equation}
	Where $F_{\tau}$ is the distribution function of $\tau$. The expected utility is
	\[
	\int_0^T (1-F_{\tau}(t)) u(\gamma_t) \, \DT.
	\]
	The result for $n=\infty$ is now easy to prove by brute force, but we wish to use symmetry.
	
	Given a measurable space with non-negative, standard measure $\mu$ we can consider the more general problem
	\begin{equation}
	\begin{aligned}
	\underset{\gamma \in L^1(\mu)}{\text{maximize}}
	& \int u(\gamma) \, \mu \\
	\text{subject to } & \gamma\geq 0 \\
	\text{and} & \int \gamma \, \mu = X_0
	\end{aligned}
	\label{eq:measureProblem}
	\end{equation}
	The optimization we wish to solve is just the special case with $\mu=(1-F_\tau(t)) \, \DT$.
	If $\mu$ is isomorphic to an invariant measure $m$ on the circle $S^1$, we can use the rotation argument above to prove that the optimizer is given by
	constant $\gamma^*$, specifically $\gamma^*= \frac{X_0}{\int \mu}$.
	
	For more general $\mu$, suppose we are given a function $\gamma$
	which satisfies the constraints. We let $\gamma \times \id$
	be the function on the space with product measure $\mu \times \lambda$ where $\lambda$ is the Lebesgue measure on $[0,1]$. But
	$\mu \times \lambda$ is an atomless, positive probability measure
	so is isomorphic to an invariant measure $m$ on the circle.
	Hence $\gamma \times \id$ is outperformed by the constant function $\gamma^*$
	on $\mu \times \lambda$. This can be projected to a constant function
	$\gamma^*$ on $\mu$. Since the objective functions on $\mu \times \lambda$ and $\mu$
	are equal for functions which are constant on the $\lambda$ factor, this means
	that $\gamma$ itself is outperformed by $\gamma^*$.
	Hence the constant function $\gamma^*$ is
	optimal on $\mu$ itself.
		
	In the case when $p_t=0$ for all times except $T-\delta t$, the
	optimization problems for different values of $n$ are all equivalent,
	so the result follows.
\end{proof}

This result provides a mathematical explanation for the intuitive appeal 
of annuities and defined benefit pensions. However, there is the obvious caveat
that the assumptions are extremely strong. One expects that
with any slight weakening of these assumptions, constant cashflows
will no longer be optimal. We prove such a converse
in the case of Epstein--Zin preferences in the paper \cite{ab-ez}.

One interpretation of this result that is worth noting is that it suggests taking the discount rate $b=0$ in our preferences
will accord better with investor expectations than choosing other values of $b$. This accords with the theoretical
discussion of discount rates in pension modelling in \cite{ab-main}.

\begin{remark}
\label{remark:omegaPerp}	
We note that the same averaging argument can be used to show that the probability space $\Omega^\perp$ 	
can be safely ignored in a complete market with concave preferences.
\end{remark}

\section{Heterogeneous Funds}
\label{sec:heterogeneous}

We now wish to consider how to manage consumption and investment when
the individuals are not identical.

For simplicity, let us assume that there are a finite number, $\ell$, of
possible initial types of investor $\{\zeta_1, \ldots, \zeta_\ell\}$.
We write ${\cal Z}$ for the set of types of investor.
Each type $\zeta$ describes the initial capital, mortality distribution
and preferences of the individual. 
Our argument will carry through with appropriate modifications
to the case of a compact space of types rather than a finite set of types. However, giving such an account would increase the technical burden to
little real purpose.

It is hypothetically appealing to identify the optimal investment strategy for the heterogeneous fund. However, we would need to define optimality, and 
this will require us to specify preferences over outcomes between different types of investor.
Describing such preferences seems challenging and the choice of
preferences would be controversial.

Instead we will take an axiomatic approach. Rather than  find
an optimal scheme for managing a fund we merely seek an {\em acceptable}
management scheme. We will define the notion of acceptable
axiomatically, and will show that for large funds all acceptable
management schemes yield similar outcomes for the investors.

We assume that the preferences
of an individual of type $\zeta$ are given by a concave, invariant gain function
${\cal J}_\zeta(\gamma,\tau)$. 
Let us write $v(n,\zeta)$ for the value function for individuals
of type $\zeta$ investing in this market in a collective of size
$n$
as modelled in Section \ref{sec:homogeneousModel}.

We assume that
\[
\lim_{n\to \infty} v(n,\zeta) = v(\infty,\zeta).
\]
Sufficient conditions to ensure this are described in the papers \cite{ab-exponential} and \cite{ab-ez}.

The initial population is determined
by a vector $\bm{\zeta} \in {\cal Z}^n$, where
component $i$ of $\bm{\zeta}$ denotes the type of the $i$-th
individual. 
A {\em management scheme} ${\cal M}$ is a function acting on
vectors $\bm{\zeta}$ of arbitrary length $n$
and which yields a strategy
$(\bm{\gamma}, \bm{\alpha})$
where $\bm{\gamma}_t$ is a vector of consumptions
of length $n$ with $i$-th component being the consumption
of individual $i$, and $\bm{\alpha}$ is an investment
strategy. We require that the combined consumption and investment are
self-financing. Thus
\[
{\cal M}: \bigsqcup_{n=1}^\infty {\cal Z}^n
\to 
\bigsqcup_{n=1}^\infty L^0(\R^n \times \R^k).
\]
and ${\cal M}$ maps the $i$-th component of the first union
into the $i$-th component of the second union.

Our first axiom for ${\cal M}$ is one of fairness, which as we have seen is a property
that arises automatically in any concave maximization problem.
\begin{axiom}{I1}
	\label{axiom:fairness}
	All surviving individuals of type $\zeta$ consume the same amount.
	That is
	\begin{enumerate}[(i)]
		\item 
		if ${\cal M}(\bm{\zeta})=(\bm{\gamma}, \bm{\alpha})$
		then
		$\gamma^i_t = \gamma^j_t$ when $\bm{\zeta}(i)=\bm{\zeta}(j)$, 
		$t\leq \tau_i$ and 	$t\leq \tau_j$.
		\item  If $\bm{\zeta}^\prime$ is obtained by permuting the elements of $\bm{\zeta}$, then
		${\cal M}(\bm{\zeta}^\prime)$ is obtained by the corresponding
		permutation of the consumption streams of ${\cal M}(\bm{\zeta})$
	\end{enumerate}	
\end{axiom}

Let the proportions of different individuals prevailing in the
population be given by a vector of weights $\omega(\bm{\zeta})=(\omega_{\zeta_1}, \ldots, \omega_{\zeta_\ell})$ with $0<\omega_\zeta<1$, $\omega_\zeta$ rational and
$\sum_{\zeta \in {\cal Z}} \omega_{\zeta}=1$.  Let $\lcm(\omega)$ denote the
lowest common multiple of the denominators of the fractions $\omega_i$,
so we know that the population is some integer multiple of $\lcm(\omega)$.

By Axiom \ref{axiom:fairness} we may define
\[
a_{\cal M}(\omega,n,\zeta)
\]
to be the value of the gain function achieved by an individual
of type $\zeta$ for a population of size $n$. This is defined if
$n$ is any integer multiple of $\lcm(\omega)$.

If $a_{\cal M}(\omega, n, \zeta) < v(n \, \omega_\zeta, \zeta)$
then the investment strategy for the heterogeneous fund will
not be able to attract investors of type $\zeta$ as they would be 
better off following the strategy of Section \ref{sec:homogeneousModel}.
Theorem \ref{thm:increasingValue} then suggests that this
would be to the detriment of all other investors in the
heterogeneous fund, as increasing collectivisation should always
be beneficial. These observations motivate the following
axioms.

\begin{axiom}{I2}
	A management scheme is monotone if:
	\[
	a_{\cal M}(\omega,m,\zeta) \leq a_{\cal M}(\omega,n,\zeta)
	\]
	if $m\leq n$.
	\label{axiom:montone}
\end{axiom}

\begin{axiom}{I3}
	A management scheme achieves the {\em performance standard} if
	\[
	a_{\cal M}(\omega,n,\zeta) \geq v(n \, \omega_\zeta, \zeta).
	\]
	for all $n$ and $\zeta$.
	\label{axiom:performance}
\end{axiom}

\begin{definition}
	A management scheme is {\em acceptable} if it satisfies Axioms \ref{axiom:fairness},
	\ref{axiom:montone} and \ref{axiom:performance}.
\end{definition}

By the monotonicity property, we may unambiguously define
\[
a_{\cal M}(\omega,\infty,\zeta) = \lim_{n\to \infty} a_{\cal M}(\omega,n,\zeta).
\]
By the performance standard and our assumption on the convergence of $v(n,\zeta)$ as $n\to \infty$ we see that for any acceptable management
scheme
\begin{equation}
a_{\cal M}(\omega,\infty,\zeta) \geq v(\infty, \zeta).
\label{eq:heterogeneousLowerBound}
\end{equation}

The scheme of simply grouping all individuals of a given
type together into a homogeneous fund and managing that according
to the model of Section \ref{sec:homogeneousModel} will yield
an acceptable management scheme which achieves the lower bound
\eqref{eq:heterogeneousLowerBound}. We call this the
{\em basic management scheme}.

We wish to show that the basic management scheme is asymptotically
optimal in complete
markets. The key observation is that collective investment
provides no substantive benefit for the investment problem 
in a complete market without mortality.

Let us consider how to write a collective investment problem
for $n$ individuals without mortality investing in our market.
We choose consumption $\gamma^i_t$ for individual $i$, and
overall investment proportions $\alpha_t$. The total fund value will
be given by the budget equations \eqref{eq:fundValue}, and hence
the set of admissible controls for this problem will be given
by \eqref{eq:admissibleControls}. The difference will be that the
preferences of the individual will depend only upon cashflows received
and not upon mortality, so we will suppose that for individual $i$ we have a gain function $\hat{\cal J}_i(\gamma^i_t)$ which depends only
upon the cashflows. We define the value function 
\[
\hat{v}_i := \sup_{{\cal A}_1} \hat{\cal J}_i(\gamma^i_t)
\]
where ${\cal A}_n$ is the set of acceptable admissible controls for $n$ individuals.
Although we have defined the set of admissible controls for $n$ individuals,
we do not write down an optimization problem for $n$ individuals as we
do not know what the objective should be across a heterogeneous collective.

Our next result shows that in a complete market without mortality there is no real
benefit in considering collectivised problems as, however we
select an admissible control (by solving an optimization problem or otherwise), it will never bring a substantive advantage to any
individual unless it also gives a substantive disadvantage to
some other individual. In other words, one cannot pay Paul without robbing Peter.

\begin{lemma}
	\label{lemma:collectivisationNoHelp}	
	Suppose the gain functions  $\hat{\cal J}_i(\gamma_t)$
	are concave, monotone and do not saturate\footnote{The definitions of these
		terms were given for gain functions over cashflows with mortality,
		but the corresponding definitions for $\hat{\cal J}$ should
		be obvious}
	and satisfy $\hat{\cal J}_i(\gamma_t)>-\infty$
	for positive cashflows $\gamma_t$ .
	Suppose that
	$\hat{\cal J}_i(\gamma_t) = -\infty$ whenever $\gamma_t$
	is negative on a set of finite measure.
	
	Let $i_*$ be an individual, then for any $\epsilon_1>0$
	there exists $\epsilon_2>0$ such that for any
	admissible investment consumption strategy for all the investors with $\gamma^i$ satisfying
	\begin{equation}
	\hat{\cal J}_{i_*}(\gamma^{i_*}_t) \geq \hat{v}_{i_*} + \epsilon_1
	\label{eq:optimalForEachI}
	\end{equation}
	there is an investor $i$ such that
	\[
	\hat{\cal J}_i(\gamma^i_t) \leq \hat{v}_i - \epsilon_2.
	\]
\end{lemma}
\begin{proof}
	Let us write $v_i(b)$ for the value function for individual
	$i$ as a function of their budget $b$. By Lemma \eqref{lemma:continuityOfV}
	below, $v_i$ is continuous as a function of $b$ for any
	$b > 0$. We write $B_i$ for the budget of individual $i$.
	
	In a complete market, we call a measurable non-negative cashflow $\gamma_t$
	a derivatives contract and we call the discounted $\Q$-measure expectation of $\gamma_t$ the price of this contract. This price may be infinite, which means that the cashflows cannot be super-replicated by any admissible trading strategy. If the price is finite, the contract can be replicated
	by an admissible trading strategy with initial budget given by the price.
	Note that the requirement that $\gamma_t$ is non-negative
	ensures that the price of derivative contracts is additive: if negative
	infinities were allowed as prices this would not be the case.
	
	If we write ${\cal D}_b$ for the set of derivatives contracts
	of price less than or equal to $b$, we have:
	\[
	\hat{v}_i(b) = \sup_{\gamma \in {\cal D}_b} \hat{\cal J}(\gamma).
	\]
	Here we have used our assumption that negative cashflows yield
	a value of $-\infty$ for the gain function.
	
	By the continuity of $v_{i_*}$, there is a price, $\delta$,
	such that if $\hat{\cal J}_{i_*}(\gamma^{i_*}_t)\geq v_{i_*} + \epsilon_1$,
	then the price of the derivative contract with payoff $\gamma^{i_*}_t$
	is at least $B_{i_*}+\delta$.
	
	By the monotonicity and non-saturation of the gain functions,
	together with the concavity given by Lemma \eqref{lemma:continuityOfV}
	below
	\[
	\hat{v}_i\left(B_{i}-\frac{\delta}{\ell-1}\right) < \hat{v}_i(B_i).
	\]
	Let
	\[
	\epsilon_2 = \inf_{i\neq i_*} \left\{\hat{v}_i(B_i)- \hat{v}_i\left(B_{i}-\frac{\delta}{2(\ell-1)}\right) \right\}.
	\]
	Any derivative contract with cashflows $\gamma^i$ satisfying
	\[
	\hat{\cal J}_i(\gamma^i_t) \leq \hat{v}_i - \epsilon_2.
	\]
	will cost at least $B_{i}-\frac{\delta}{2(\ell-1)}$. Hence
	the total cost of an investment strategy yielding all the 
	cashflows $\gamma^i$ is at least 
	\[
	B_{i}+\delta + \sum_{i\neq i_*} \left( B_{i}-\frac{\delta}{2(\ell-1)} \right)
	= \sum_i B_{i} + \frac{\delta}{2}
	\]
	which is greater than the total budget and hence cannot be admissible.
\end{proof}

\begin{lemma}
	\label{lemma:continuityOfV}
	Consider investment without mortality in a homogeneous market. Suppose
	an individual's gain function over consumption, $\hat{\cal J}(\gamma)$,
	is concave and monotone. Let ${\cal A}_b$ be the set of admissible 
	consumption, investment strategies with initial budget $b$.
	Define
	\[
	\hat{v}(b) = \sup_{(\gamma,\alpha)\in{\cal A}_b} \hat{\cal J}_i(\gamma).
	\]
	The function $v$ is concave and hence $v$ is
	continuous on any open set on which it is finite.
\end{lemma}
\begin{proof}
	Given two budgets $b_1$, $b_2>b$
	then let $(\gamma_i,\alpha_i)$ be admissible strategies for each
	budget. So for any $\lambda \in [0,1]$, $(\lambda \gamma_1 + (1-\lambda) \gamma_2, \lambda \alpha_1 + (1-\lambda) \alpha_2)$ is an
	admissible strategy with budget $\lambda b_1 + (1-\lambda) b_2$.
	By the concavity of $\hat{\cal J}$,
	\[
	\hat{\cal J}(\lambda \gamma_1 + (1-\lambda) \gamma_2 )
	\geq \lambda \hat{\cal J}(\gamma_1) + (1-\lambda) \hat{\cal J}(\gamma_2 ).
	\]
	Hence
	\[
	\hat{v}(\lambda b_1 + (1-\lambda) b_2 )
	\geq \lambda \hat{v}(b_1) + (1-\lambda) \hat{v}(b_2).
	\]
\end{proof}

\begin{remark}
	Lemma \ref{lemma:collectivisationNoHelp} is only true in complete
	markets. If two individuals have different risk or consumption preferences
	in an incomplete market then it can often be beneficial to design
	a derivatives contract to the mutual advantage of both parties.
	This is the essential purpose of derivatives contracts and explains
	why there is a market for such contracts.
\end{remark}

We are now ready to prove the main result of this section.
\begin{theorem}
	\label{thm:acceptableFundTheorem}
	For any acceptable management scheme in a complete market
	\begin{equation*}
	a_{\cal M}(\omega,\infty,\zeta) = v(\infty, \zeta)
	\end{equation*}
	so long as all individual gain functions are concave, 
	monotone, invariant and do no saturate, and so long as
	\begin{equation}
	\lim_{n \to \infty} v(n,\zeta) = v(\infty, \zeta)>-\infty.
	\label{eq:convergenceRequirement}
	\end{equation}
	for any positive budget.
\end{theorem}
\begin{proof}
    We will prove the result in the case of discrete time consumption. The case of continuous
    time consumption is similar.

	By Axiom \ref{axiom:performance},  $\lim_{n \to \infty} a_{\cal M}(\omega, n, \zeta) \geq \lim_{n \to \infty} v(n \omega^\zeta, \zeta)$. By assumption  
	$\lim_{n \to \infty} v(n, \zeta)= v(\infty,\zeta)$. Hence $a(\omega, \infty, \zeta) = \lim_{n \to \infty} a_{\cal M}(\omega, n, \zeta) \geq v(\infty, \zeta)$.

	Let us now define the meaning of an admissible strategy for an infinite heterogeneous
	collective where the types of each individual are given by the
	proportions $\omega$. We will
	assume that surviving individuals of a given
	type, $\zeta$, all consume $\gamma^\zeta_t$ at time $t$.
	Hence the total consumption per person (i.e.\ of any type, including both survivors and the deceased) is
	\[
	\sum_{\zeta \in {\cal Z}} \gamma^{\zeta}_t \pi^{\zeta}_t \omega^{\zeta}
	\]
	where $\pi^{\zeta}_t$ denotes the proportion of individuals of type $\zeta$
	who survive to time $t$. Note that $\pi^{\zeta}_t=1-F_{\tau_\zeta}(t)$
	where ${\tau}_\zeta$ is a random variable distributed according
	to the time-of-death distribution for individuals of type $\zeta$.
	
	If we let $Y_t$ denote the fund value per person before consumption and
	$\overline{Y}_t$ denote the fund value per person after consumption we
	have the following budget equations for the dynamics of $Y_t$ and $\overline{Y}_t$. Let us write $B^\zeta$ for the initial
	budget of individuals of type $\zeta$. Then the fund value
	per person for the infinite heterogenous collective should be defined
	to follow the dynamics
	\begin{equation}
	\begin{split}
	Y_t &= \begin{cases}
	\sum_{\zeta \in {\cal Z}} \omega^\zeta B^\zeta & \text{t = 0} \\
	\lim_{h\to 0+}\overline{Y}_{t-h} &  t \in {\cal T}\setminus\{0\} \\
	\overline{Y}_t & \text{otherwise.}
	\end{cases} \\
	\overline{Y}_t &=
	\begin{cases}
	Y_t - \sum_{\zeta \in {\cal Z}} \gamma^{\zeta}_t \pi^{\zeta}_t \omega^{\zeta}
	& t \in {\cal T} \\
	\overline{Y}_{t^\prime}  + \sum_{i=1}^k \int_{t^\prime}^t \alpha^i_s \overline{Y}_s \, \ed S^i_s
	& t^\prime \in {\cal T} \text{ and } t^\prime\leq t < t^\prime+\delta t.
	\end{cases}
	\end{split}
	\label{eq:budgetHeterogeneousCollective}
	\end{equation}
	
	We may alternatively view the equations above as describing the dynamics of a fund
	of ``virtual individuals'' where each virtual individual represents
	the interests of an infinite fund of individuals all of type $\zeta$.
	To see this, we define $\hat{\gamma}^\zeta_t=\gamma^{\zeta}_t \pi^{\zeta}_t \omega^{\zeta}$. We say that	
	virtual individual $\zeta$ consumes an amount $\hat{\gamma}^\zeta_t$ at each time $t \in {\cal T}$. We say that the initial budget
	of virtual individual $\zeta$ is $\omega^\zeta B^\zeta$. We may now view
	equation \eqref{eq:budgetHeterogeneousCollective} as giving
	the dynamics of the total fund value before consumption $Y_t$
	of a collective of these heterogeneous virtual individuals.
	This observation will allow us to apply Lemma
	\ref{lemma:collectivisationNoHelp} to the collective investment problem
	with mortality.
	
	Let us define a gain function (without mortality) $\hat{\cal J}_\zeta$ for 
	the virtual individual $\zeta$ by
	\[
	\hat{\cal J}_\zeta(\hat{\gamma}^\zeta_t):={\cal J}_\zeta( \gamma^{\zeta}_t, {\tau}_\zeta)
	\]
	
	Let us write $\hat{v}_\zeta(b)$ for the value function of a virtual individual with this gain function with an initial budget $b$. We see that $\hat{v}_\zeta( \omega^\zeta B_\zeta)=v(\infty, \zeta)$.
	
	Suppose $(\bm{\gamma},\bm{\alpha})$ is a strategy for the infinite
	heterogeneous collective
	and let $\zeta_*$ be a chosen type of individual.
	It follows by Lemma \ref{lemma:collectivisationNoHelp} that 
	for any given $\epsilon_1>0$
	such that 
	\begin{equation}
	\hat{\cal J}_{\zeta_*}(\hat{\gamma}^{\zeta_*}_t) \geq v(\infty, \zeta_*) + \epsilon_1
	\end{equation}
	there exists $\epsilon_2>0$ and a type $\zeta$ such that
	\begin{equation}
	\hat{\cal J}_{\zeta}(\hat{\gamma}^\zeta_t) \leq v(\infty, \zeta) - \epsilon_2.
	\end{equation}	
	Hence given any $\epsilon_1>0$
	such that 
	\begin{equation}
	{\cal J}_{\zeta_*}(\gamma^{\zeta_*}_t, \tau_{\zeta_*}) \geq v(\infty, \zeta_*) + \epsilon_1
	\label{eq:sigmaStarCondition}
	\end{equation}
	there exists $\epsilon_2>0$ and a type $\zeta$ such that
	\begin{equation}
	{\cal J}_{\zeta}(\gamma^\zeta_t,\tau_{\zeta}) \leq v(\infty, \zeta) - \epsilon_2.
	\label{eq:epsilon2Condition}
	\end{equation}
	
	Now let us suppose for a contradiction that for some $\epsilon_1>0$,
	$\zeta_*$ and $N$ we have $a(\omega, N, \zeta_*) \geq v(\infty, \zeta_*) + 2 \epsilon_1$. Note that we then have $a(\omega, n, \zeta_*) \geq v(\infty, \zeta_*) + 2 \epsilon_1$ for all $n \geq N$. We may then choose $\epsilon_2>0$ as described in the preceding paragraph.

	There are only finitely many types $\zeta$ and we know $\lim_{n\to \infty} a_{\cal M}(\omega,n,\zeta)\geq  v(\infty, \zeta)$. So for sufficiently large $n$ we may assume that 
	\[
	a_{\cal M}(\omega, n, \zeta) \geq v(\infty, \zeta)-\tfrac{1}{3} \epsilon_2.
	\]
	Hence for such $n$ we may find an investment-consumption strategy $(\bm{\gamma}, \bm{\alpha})$ for a collective of $n$ individuals
	which yields the same consumption for all surviving individuals
	of a given type and which satisfies
	\begin{equation*}
	{\cal J}_{\zeta_*}(\gamma^{\zeta_*}_t, \tau_{\zeta^*}) \geq v(\infty, \zeta_*) + \epsilon_1
	\end{equation*}
	and which also satisfies
	\begin{equation*}
	{\cal J}_{\zeta}(\gamma^\zeta_t, _{\zeta}) \geq v(\infty, \zeta) - \tfrac{2}{3} \epsilon_2
	\end{equation*}
	for all types $\zeta$.
	
	Since the discounted asset prices $S^i_t$ are $\Q$-martingales
	we have
	\[
	\E_{\Q \times \P^L} 
	\left( \sum_{\zeta \in {\cal Z}} 
	\int e^{-rt} \omega^\zeta \gamma_t^\zeta \, \DT \right)
	\leq \sum_{\zeta \in {\cal Z}} \omega^\zeta B^\zeta.
	\]
	Since the consumption $\gamma_t$ is non-negative we have by
	Fubini's theorem that
	\[
	\overline{\gamma}^\zeta_t := \E_{\P^L} 
	\left( 
	\int e^{-rt} \omega^\zeta \gamma_t^\zeta \, \DT \mid \bm{S}_t \right)
	\]
	is a progressively measurable process for the filtered probability
	space $(\Omega^M, {\cal F}_t, \P^M)$ and satisfies
	\[
	\sum_{\zeta \in {\cal Z}} \E_{\Q} \left(
	\int e^{-rt} \omega^\zeta \overline{\gamma}_t^\zeta \, \DT \right)
	\leq \sum_{\zeta \in {\cal Z}} \omega^\zeta B^\zeta.
	\]
	By the complete market assumption we can then find
	an investment strategy $\overline{\bm{\alpha}}$ that funds
	$\overline{\bm{\gamma}}$. Hence we have found an
	admissible strategy 
	$(\overline{\bm{\gamma}},\overline{\bm{\alpha}})$
	for the infinite homogeneous collective which by
	the concavity of the preferences will satisfy
	\begin{equation*}
	{\cal J}_{\zeta_*}(\overline{\gamma}^{\zeta_*}_t, \tau_{\zeta^*}) \geq v(\infty, \zeta_*) + \epsilon_1
	\end{equation*}
	and which also satisfies
	\begin{equation*}
	{\cal J}_{\zeta}(\overline{\gamma}^\zeta_t, _{\zeta}) \geq v(\infty, \zeta) - \tfrac{2}{3} \epsilon_2
	\end{equation*}
	for all types $\zeta$. This contradicts equations \eqref{eq:sigmaStarCondition} and \eqref{eq:epsilon2Condition}.
	
	We deduce that $a(\omega, N, \zeta) \leq v(\infty, \zeta)$ for
	all $\zeta$, which gives the result.
\end{proof}

The financial significance of Theorem \ref{thm:acceptableFundTheorem}
is that all acceptable strategies are asymptotically
equivalent in a complete market, in particular this applies to the
Black--Scholes--Merton market.

The result is analogous to the classical result that
in the Black--Scholes--Merton market any derivative has a unique price
independent of the preferences of the investor. This too is an asymptotic result, in the sense that it assumes continuous time
and zero transaction costs.

In practice no two individuals are alike, and
so our assumption of a finite number of distinct types can
be criticised.
However, simple modifications of our strategy for heterogeneous funds will yield good
results if there are a large number of similar individuals. We give a concrete algorithm in \cite{ab-main}
and show that it achieves approximately $98\%$ of the maximum benefit of collectivisation for a heterogeneous
fund of only $100$ investors.

\begin{remark}
It is natural to ask how our results extend if one incorporates systematic longevity risk.
The simplest approach would be to assume that there is a complete market of contracts
on the risk factors determining systematic longevity risk, in which case a similar result
can be expected to hold. Longevity derivatives do exist, but at present the market
is not very liquid. Our results suggest that each type of individual will have
different views on the attractiveness of longevity risk, which would imply that the
development of collective funds should result in a more liquid market in longevity
derivatives.
\end{remark}

\appendix

\section{Proof of Theorem \ref{thm:ezConvergence}}
\label{sec:ezConvergence}

\begin{proof}[Proof of Theorem \ref{thm:ezConvergence}]
Recall that when working with Epstein--Zin preferences we assume $\Omega^\perp$ is trivial.

To prove the existence of Epstein--Zin utility with mortality following the methods of
\cite{xing} one first makes the change of variables
$(\tilde{V}_t,\tilde{Z}_t, \tilde{\zeta}_t)=\alpha e^{-\frac{b \alpha t}{\rho}} (V_T,Z_t,\zeta_t)$ 
to transform \ref{eq:bsde1} to the BSDE
\begin{equation}
\ed \tilde{V}_t =
F(t,\gamma_t, \tilde{V}_t) {\1}_{t \leq \tau} \, \ed t
- \tilde{Z}_t \, \ed W_t
- \sum_{i=1}^n \tilde{\zeta}_t^i \, \ed M^i_t, \quad 0 \leq t \leq T; \quad \tilde{V}_T = e^{ -\frac{bT\alpha}{\rho}} \alpha U
\label{eq:bsde2}
\end{equation}
where
\[
F(t,\gamma_t,v):= \frac{b \alpha \exp( -b t )}{\rho} \gamma_t^\rho v^{1-\frac{\rho}{\alpha}}.
\]
To prove the existence and uniqueness of this BSDE, one next solves the family of BSDEs indexed by $m\in \R_{\geq 0}$,
\begin{equation}
\ed \tilde{V}^m_t =
F^m(t,\gamma_t, \tilde{V}^m_t) {\1}_{t \leq \tau} \, \ed t
- \tilde{Z}_t \, \ed W_t
- \sum_{i=1}^n \tilde{\zeta}_t^i \, \ed M^i_t, \quad 0 \leq t \leq T; \quad \tilde{V}_T = e^{ -\frac{bT\alpha}{\rho}} \alpha U
\label{eq:bsde3}
\end{equation}
where
\[
F^m(t,\gamma_t,v):=b \frac{\alpha}{\rho} \exp( -b t )(\gamma_t^\rho \wedge m)
(v \wedge m)^{1-\frac{\rho}{\alpha}}.
\]
The driver $F^m$ has the property that $y \to F^m(t,\gamma_t,y)$ is Lipschitz (since
our parameter assumptions ensure $1-\frac{\rho}{\alpha}>1$)
and so one can use the theory of \cite{dumitrescu} to obtain the existence of a unique
solution to this BSDE. The comparison theorem ensures that $\tilde{V}^m$ is
non-negative and decreasing in $m$. The proof of Proposition 2.2 in \cite{xing} then
shows that the $\tilde{V}=\lim_{m \to \infty}\tilde{V}^m$ can be extended to give
a solution to \eqref{eq:bsde2}.

Let us define $\EZ^{m}(\gamma,\tau):=\frac{1}{\alpha}\tilde{V}^m_0$. We have that
$EZ^m(\gamma, \tau)$ is finite, non-positive and increasing in $m$ with limit given
by the Epstein--Zin utility with mortality which we denote $\EZ(\gamma,\tau)$.

\smallskip

Let $M < v_\infty$ and $\epsilon>0$ be given.
We may find an admissible consumption stream $\gamma^\infty_t$ for the
problem of an infinite
collective with initial budget $B$ such that $\EZ(\gamma^\infty_t, \tau^\iota) \geq M$. 

Since $\EZ$ is concave and finite, for any admissible $\gamma$,
the map $\lambda\to \EZ( \lambda \gamma^\infty_t, \tau^\iota)$ is continuous on $\R_{\geq 0}$.
Hence we may choose $\lambda\in(0,1)$ such that $\EZ(\lambda \gamma^\infty_t, \tau^\iota) \geq M-\epsilon$. By the convergence of $\EZ^m$ in $m$, we may then find $m$ such that
\begin{equation}
\EZ^m(\lambda \gamma^\infty_t, \tau^\iota) \geq M-2 \epsilon.
\label{eq:choiceOfM}
\end{equation}

Suppose that we have an initial budget of $B$ for the problem of a collective of
$n$ investors. Recall that $n_t$ denotes the number of survivors at time $t$.
Let $G_{n,t}$ be the event defined by
\[
G_{n,t} = \{\omega \mid n_s(\omega)\leq \frac{1}{\lambda} \E(n_s) \, \forall 0\leq s \leq t\}.
\]
We write $\lnot G_{n,t}$ for its compliment.
The consumption stream $\lambda \1_{G_{n,t}} \gamma_t^\infty$ will be admissible, as the consumption per survivor will always be less than $\gamma_t$.

Define $F^{m,\infty}(t,v):=F^m(t,\lambda \gamma_t^\infty,v)$ and
$F^{m,n}(t,v):=F^m(t,\lambda \1_{G_{n,t}} \gamma_t^\infty,v)$ for finite $n$. Write
$(\tilde{V}^{m,n},\tilde{Z}^{m,n},\tilde{\zeta}^{m,n})$
for the solution to \eqref{eq:bsde3} for the driver $F^{m,n}$ for both
finite and infinite values of $n$, and define $\EZ^{m,n}:=\frac{1}{\alpha}\tilde{V}^{m,n}$

Differentiating the expression for $F^{m,n}(t,v)$ with respect to $v$,
we obtain
the following Lipschitz estimate for the drivers
\[
|F^{m,n}(t,y_1)-F^{m,n}(t,y_2)| \leq \left( 1 - \frac{\rho}{\alpha} \right) b \frac{|\alpha|}{\rho} m^{1-\frac{\rho}{\alpha}} |y_1-y_2|=:C_m |y_1-y_2|
\]
for an appropriately defined constant $C_m$.

Let us define $e=|\tilde{V}^{m,\infty}-\tilde{V}^{m,n}|$.
By Proposition 2.4 of \cite{dumitrescu}, we obtain the following bound on $e$:
\[
e^2
\leq \eta_m \E [ \int_0^T e^{\beta_m s} |F^{m,\infty}(s,\tilde{V}^{m,\infty})
- F^{m,n}(s,\tilde{V}^{m,n}) |^2 \, \ed s ],
\]
where $\eta_m:=\frac{1}{C_m^2}$ and $\beta_m:=3 C_m^2 + 2 C_m$. Inserting
the definitions of $F^{m,n}$ we obtain
\[
e^2 \leq \eta_m \E[ 
\int_0^T  \frac{b \alpha}{\rho} e^{\beta_m s} ((\lambda \1_{\lnot G_{n,s}} \gamma^\infty_s)^\rho \wedge m)(\tilde{V}^{m,\infty} \wedge m)^{1-\frac{\rho}{\alpha}} \, \ed s.
\]
Splitting the integral into two regions, we find that for $\delta\in(0,T)$,
\begin{align}
e^2 &\leq \eta_m  \frac{ b \alpha m^{2-\frac{\rho}{\alpha}}}{\rho} e^{\beta_m T} 
\E[ \int_0^{T-\delta} \1_{\lnot G_{n,s}} \, \ed s + \int_{T-\delta}^T \ed s ] \nonumber \\
&\leq \eta_m  \frac{ b \alpha m^{2-\frac{\rho}{\alpha}}}{\rho} e^{\beta_m T} 
\E[ \int_0^{T-\delta} \1_{\lnot G_{n,T-\delta}} \, \ed s + \delta ] \quad \text{as }\lnot G_{n,s} \subseteq \lnot G_{n,t} \text{ if } s\leq t \nonumber \\
&\leq \tilde{C}_m (T
\P( \lnot G_{n,T-\delta} ) + \delta ) \label{eq:eBound}
\end{align}
for an appropriately defined $C_m$. Choose $\delta$ such that
$\tilde{C}_m \delta \leq \frac{1}{2}\alpha^2 \epsilon^2$. Now use the
fact that $P(\lnot G_{n,T-\delta})\to 0$ as $n \to \infty$ (as follows readily from Lemma \ref{lemma:finiteTimePoints}, below) to choose
$n$ such that $\tilde{C}_m (T
\P( \lnot G_{n,T-\delta} ) \leq \frac{1}{2}\alpha^2 \epsilon^2$. By \eqref{eq:eBound} we will then have
that $e \leq |\alpha| \epsilon$. We then have
\[
v_n \geq \EZ( \lambda \1_{G_{n,t}} \gamma^\infty_t ) \geq |\EZ^{m,n}|
\geq |\EZ^{m,\infty}| - \frac{e}{|\alpha|} \geq |\EZ^{m,\infty}| - \epsilon \geq M - 3\epsilon.
\]
where we have used in sequence: the definition of $v_n$; the convergence of the increasing sequence in $m$ given by $\EZ^{m,n}$; the definitions of $e$ and $\EZ^{m,n}$; our bound for $e$; equation \eqref{eq:choiceOfM}.
So by Theorem \ref{thm:increasingValue}, $\lim_{n \to \infty}v_n=v_\infty$.
\end{proof}

We now prove the Lemma used in the proof.

\begin{lemma}
    Let $T^*$ be minimum time by which
    an individual is almost sure to have
    died.
	For a continuous mortality distribution,
	given a time point $0\leq t_0<T^*$, for any $\epsilon \in (0,1)$ there exists
	a finite set of points $t_i \in [0,t_0)$, indexed by $i \in I$ such that
	\[
	\P\left( \forall t \in [0,t_0] \,:\, n_t \leq \left( \frac{1}{1-\epsilon} \right)^2 \E(n_t) \right) 
	\geq 
	\P\left( \forall i \in I \,:\, n_{t_i} \leq \left( \frac{1}{1-\epsilon} \right) \E(n_{t_i}) \right)
	.
	\]
	\label{lemma:finiteTimePoints}
\end{lemma}
\begin{proof}
	We define $t_i$ inductively. If $t_{i-1}=0$, we are done and take the index 
	set $I=\{0,1,\ldots,i-1\}$. Otherwise define
	\[
	t_i = \inf \left\{ t \mid t=0 \text{ or } \E(n_{t}) \leq \frac{1}{1-\epsilon} \E(n_{t_{i-1}}) \right\}.
	\]
	If $t_i\neq 0$ we see $\E(n_{t_{i-1}}) \geq \frac{1}{1-\epsilon} \E(n_{t_i})$,
	so for sufficiently large $i$ we must have $t_i=0$. Hence the index set, $I$, is finite.
	
	Given $t \in [0,t_0]$ we can find $i \in I$ with $t_i \leq t \leq t_{i-1}$.
	Suppose 
	\[
	n_t > \left( \frac{1}{1-\epsilon}\right)^2 \E( n_t)	
	\]
	then we have
	\[
	n_{t} > \left( \frac{1}{1-\epsilon}\right)^2 \E( n_{t_i})
	\geq \left( \frac{1}{1-\epsilon}\right)^2 \E( n_{t_{i-1}})
	\geq \left( \frac{1}{1-\epsilon}\right) \E( n_{t}).
	\]
\end{proof}

\bibliography{collectivization}

\begin{thebibliography}{10}

\bibitem{armstrongClassification}
John Armstrong.
\newblock Classifying markets up to isomorphism.
\newblock {\em arXiv preprint arXiv:1810.03546}, 2018.

\bibitem{ab-exponential}
John Armstrong and Cristin Buescu.
\newblock Collectivised pension investment with exponential
  {K}ihlstrom--{M}irman preferences.
\newblock {\em arXiv preprint arXiv:1911.02296}, 2019.

\bibitem{ab-ez}
John Armstrong and Cristin Buescu.
\newblock Collectivised pension investment with homogeneous {E}pstein--{Z}in
  preferences.
\newblock {\em arXiv preprint arXiv:1911.10047}, 2019.

\bibitem{ab-main}
John Armstrong and Cristin Buescu.
\newblock Collectivised post-retirement investment.
\newblock {\em arXiv preprint arXiv:1909.12730}, 2019.

\bibitem{aurandHuang}
Joshua Aurand and Yu-Jui Huang.
\newblock {E}pstein-{Z}in utility maximization on random horizons.
\newblock {\em arXiv preprint arXiv:1903.08782}, 2019.

\bibitem{bansal}
Ravi Bansal.
\newblock Long-run risks and financial markets.
\newblock Technical report, National Bureau of Economic Research Cambridge,
  Mass., USA, 2007.

\bibitem{bansalYaron}
Ravi Bansal and Amir Yaron.
\newblock Risks for the long run: A potential resolution of asset pricing
  puzzles.
\newblock {\em The Journal of Finance}, 59(4):1481--1509, 2004.

\bibitem{benzoniEtAl}
Luca Benzoni, Pierre Collin-Dufresne, and Robert~S Goldstein.
\newblock Explaining asset pricing puzzles associated with the 1987 market
  crash.
\newblock {\em Journal of Financial Economics}, 101(3):552--573, 2011.

\bibitem{bhamraEtAl}
Harjoat~S Bhamra, Lars-Alexander Kuehn, and Ilya~A Strebulaev.
\newblock The levered equity risk premium and credit spreads: A unified
  framework.
\newblock {\em The Review of Financial Studies}, 23(2):645--703, 2009.

\bibitem{campbellViceira}
John~Y Campbell and Luis~M Viceira.
\newblock Consumption and portfolio decisions when expected returns are time
  varying.
\newblock {\em The Quarterly Journal of Economics}, 114(2):433--495, 1999.

\bibitem{duffieEpstein}
Darrell Duffie and Larry~G Epstein.
\newblock Stochastic differential utility.
\newblock {\em Econometrica: Journal of the Econometric Society}, pages
  353--394, 1992.

\bibitem{dumitrescu}
Roxana Dumitrescu, Miryana Grigorova, Marie-Claire Quenez, and Agn{\`e}s Sulem.
\newblock {BSDE}s with default jump.
\newblock In {\em The Abel Symposium}, pages 233--263. Springer, 2016.

\bibitem{epsteinZin1}
Larry~G. Epstein and Stanley~E. Zin.
\newblock Substitution, risk aversion, and the temporal behavior of consumption
  and asset returns: A theoretical framework.
\newblock {\em Econometrica}, 57(4):937--969, 1989.

\bibitem{jeanblancEtAl}
Monique Jeanblanc, Marc Yor, and Marc Chesney.
\newblock {\em Mathematical methods for financial markets}.
\newblock Springer Science \& Business Media, 2009.

\bibitem{krepsPorteus}
David~M Kreps and Evan~L Porteus.
\newblock Temporal resolution of uncertainty and dynamic choice theory.
\newblock {\em Econometrica}, pages 185--200, 1978.

\bibitem{merton1969lifetime}
Robert~C Merton.
\newblock Lifetime portfolio selection under uncertainty: The continuous-time
  case.
\newblock {\em The Review of Economics and Statistics}, pages 247--257, 1969.

\bibitem{rawls}
John Rawls.
\newblock {\em A Theory of Justice}.
\newblock Harvard University Press, 1971.

\bibitem{xing}
Hao Xing.
\newblock Consumption--investment optimization with {E}pstein--{Z}in utility in
  incomplete markets.
\newblock {\em Finance and Stochastics}, 21(1):227--262, 2017.

\end{thebibliography}
\bibliographystyle{plain}

\end{document}